\documentclass[11pt]{article}
\usepackage[margin=1in]{geometry}
\usepackage[utf8]{inputenc}
\usepackage{graphicx,xcolor,xspace}
\usepackage{pifont}
\usepackage{mathpazo}
\usepackage{amsthm,amsmath,amssymb,mathtools}
\usepackage{tcolorbox}
\tcbuselibrary{skins,breakable}
\tcbset{enhanced jigsaw}
\usepackage[pagebackref=true,colorlinks,citecolor=olive,
	bookmarks,bookmarksopen,bookmarksnumbered]
	{hyperref}
\usepackage[noabbrev,nameinlink]{cleveref}
\newtheorem{theorem}{Theorem}
\newtheorem{lemma}{Lemma}[section]
\newtheorem{definition}[lemma]{Definition}
\newtheorem{observation}[lemma]{Observation}
\allowdisplaybreaks

\title{A Deterministic Near-Linear Time Approximation Algorithm\\for Restricted Shortest Paths in Directed Graphs}
\author{Vikrant Ashvinkumar \thanks{INSAIT, Sofia University ``St. Kliment Ohridski''}}
\date{}
\newcommand{\bracket}[1]{\left[#1\right]}

\newcommand{\card}[1]{\left\vert{#1}\right\vert}
\newcommand{\set}[1]{\left\{ #1 \right\}}

\newcommand{\IR}{\ensuremath{\mathbb{R}}}

\newcommand{\IN}{\ensuremath{\mathbb{N}}}
\newcommand{\prob}[1]{\mathbb{P}\bracket{#1}}
\newcommand{\expect}[1]{\Exp\bracket{#1}}
\DeclareMathOperator*{\Exp}{\ensuremath{{\mathbb{E}}}}
\newcommand{\algname}[1]{\underline{\textbf{\textsc{#1}}}\\}
\newenvironment{tbox}{\begin{tcolorbox}[
		enlarge top by=5pt,
		enlarge bottom by=5pt,
		 breakable,
		 boxsep=2pt,
                  left=5pt,
                  right=7pt,
                  top=10pt,
                  arc=0pt,
                  boxrule=1pt,toprule=1pt,
                  colback=white
                  ]%
	}
{\end{tcolorbox}}
\newcommand{\eps}{\ensuremath{\varepsilon}}
\newcommand{\Ot}{\ensuremath{\widetilde{O}}}
\newcommand{\cQ}{\mathcal{Q}}
\newcommand{\rsp}{\mathsf{RSP}}
\newcommand{\pois}{\textrm{Pois}\xspace}
\newcommand{\expo}{\textrm{exp}\xspace}
\newcommand{\expod}{\textrm{Exp}\xspace}
\newcommand{\lclock}{X_{\textrm{length}}\xspace}
\newcommand{\cclock}{X_{\textrm{cost}}\xspace}
\newcommand{\llambda}{\lambda_{\textrm{length}}\xspace}
\newcommand{\clambda}{\lambda_{\textrm{cost}}\xspace}
\newcommand{\dpt}{\textrm{DP}\xspace}
\newcommand{\lclockt}{Q_{\textrm{length}}\xspace}
\newcommand{\cclockt}{Q_{\textrm{cost}}\xspace}
\newcommand{\lH}{H_{\textrm{length}}\xspace}
\newcommand{\cH}{H_{\textrm{cost}}\xspace}
\newcommand{\explode}{\textcolor{orange}{\ding{89}}\xspace}

\begin{document}
\maketitle
\begin{abstract}
    In the restricted shortest paths problem, we are given an $n$-vertex $m$-edge graph $G$ whose edges are assigned two non-negative weights (or resources) called lengths and costs, a source vertex $s$, and a cost threshold $C$.
    The goal is to find, for each target vertex $t$, the length of the shortest $(s,t)$-path whose total cost is at most $C$.

    \cite{ashvinkumar2025faster} gave the first randomized bicriteria $(1+\eps, 1+\eps)$-approximate algorithm on \emph{directed} graphs running in sub-$mn$ time, answering a problem posed by \cite{Bernstein12}.
    We substantially improve this on two fronts, supplying a \emph{deterministic} bicriteria $(1+\eps, 1+\eps)$-approximate algorithm running in \emph{near-linear} time.

    Our approach straightforwardly generalizes to more than two resources.
    For $k$ resources, a deterministic $\underbrace{(1+\eps, \ldots, 1+\eps)}_{k\ \mathrm{times}}$-approximate algorithm running in \emph{near-linear} time is possible.
\end{abstract}

\textbf{AI Statement. }
Some version of \Cref{thm:main} was elicited through numerous discussions with GPT-6 Astra and older GPT models, under human guidance:
I suggested using the two Poisson clocks, and generalizing the $(2+\eps,2)$-approximate algorithm with them; the machine fleshed out details, and also supplied the derandomization.
I bear all responsibilities that attend this result.

\section{Introduction}

The restricted shortest paths ($\rsp$) problem puts a natural and practically relevant spin on a truly classical problem: shortest paths.
In the shortest paths problem, we ask: how quickly can we get from $s$ to $t$?
This is well studied and, to any non-specialist, solved for all intents and purposes.\footnote{The cognoscenti, however, would see an abundance of open problems spanning negative weights \cite{BernsteinNW22}, negative real weights \cite{fineman2024single}, universal optimality \cite{haeupler2024universal}, breaking sorting barriers \cite{duan2025breaking}, etc.}
The simple twist is to ask: how quickly \emph{and cheaply} can we get from $s$ to $t$?

In the $\rsp$ problem, we are given a graph $G=(V,E)$ equipped with \emph{exactly two} non-negative functions $\ell, c$ assigning two unrelated weights (or resources) -- called \emph{lengths} and \emph{costs} --  to the edges of $G$.
For a source-target pair $s,t\in V$, and a cost threshold $C$, the goal is to find the shortest length $(s,t)$-path in $G$ whose total cost is no more than $C$.
While it is barely any harder a problem to state than shortest paths, $\rsp$ models real-world concerns much more potently, with applications spanning quality of service (QoS) routing~\cite{lorenz1998qos, younis2003constraint}, scheduling~\cite{BriskornCLLP10, NaorST07}, and other areas of operations research (see, e.g.,~\cite{GarroppoGT10}).

More generally, the paths we look for may be subject to numerous constraints.
For example, when planning a long journey, there are many factors to take into account: the total price, number of transfers, trip length, carbon footprint, etc.
Problems of this kind are collectively known as \emph{multi-constrained path} problems (see, e.g., \cite{GarroppoGT10, XueZTT08}) or, in this paper, $\rsp$ with $k$ resources when there is one objective and $k-1$ constraints, all of which must be simultaneously satisfied.

Unlike shortest paths, $\rsp$ is known to be NP-hard~\cite{GareyJ79}, which motivates looking for approximate algorithms.
Even with this concession, the complexity of approximation algorithms for $\rsp$ in its full generality is far from being settled~\cite{ashvinkumar2025faster,Bernstein12, goel2001efficient, Hassin92, lorenz2001simple, Warburton87}.

\subsection*{Prior Work}

To discuss the previous work in more detail, let us first introduce some notation.

\begin{definition}
Let $P^*$ be the shortest $(s,t)$-path with cost bounded by $C$ in $G$.
We say that an $(s,t)$-path $P \subseteq G$ constitutes an \emph{$(\alpha,\beta)$-approximate} solution to the $\rsp$ problem, if its length satisfies $\ell(P)\leq \alpha\cdot \ell(P^*)$ and its cost satisfies $c(P)\leq \beta\cdot C$.

We call an $\rsp$ algorithm $(\alpha,\beta)$-approximate if it produces $(\alpha,\beta)$-approximate solutions.
\end{definition}

\cite{goel2001efficient} gave an $\Ot_\eps(mn)$-time\footnote{Throughout this paper, we use the $\Ot(\cdot)$ notation to suppress factors polylogarithmic in $n$, the $O_\eps(\cdot)$ notation to suppress factors in $\eps^{-1}$, and a combination of the two notations.} $(1,1+\eps)$-approximate algorithm for the single-source $\rsp$ problem.
Their result leverages the observation that simple dynamic programming can, in fact, solve any instance of $\rsp$ with integral costs in $O(mC)$ time~\cite{joksch1966shortest}, and combines it with cost rounding/scaling in a standard way.
\cite{lorenz2001simple} showed a $(1+\eps,1)$-approximate algorithm for the single-source single-target case running in $\Ot_\eps(mn)$ time; this essentially follows by swapping the roles of lengths and costs and proceeding similarly\footnote{The algorithm of \cite{lorenz2001simple} requires some extra care to make it run in strongly polynomial time. This is also the reason for the single-target requirement.}
to \cite{goel2001efficient}.  

In the special case of \emph{undirected} graph inputs, \cite{Bernstein12} obtained a randomized $(1+\eps,1+\eps)$-approximate single-source $\rsp$ algorithm running in $O_\eps(1)m^{1+o(1)}$ time.
Apart from the rounding/scaling DP, \cite{Bernstein12} exploits an emulator construction that crucially relies on the input graph being undirected.

Recently, \cite{ashvinkumar2025faster} showed randomized $(1+\eps,1+\eps)$-approximate single-source $\rsp$ algorithms running in $\Ot_\eps(n^2 \log(nW))$ time -- near optimal for dense graphs -- and $\Ot_\eps(mn^{3/5} \log(nW))$ time.
Here, $W$ is the aspect ratio of lengths.\footnote{Formally, the aspect ratio $W$ is the quantity $(\max_{uv\in E_{\ell+}}\ell(uv))/(\min_{uv\in E_{\ell+}}\ell(uv))$, where $E_{\ell+}$ is the subset of edges with positive lengths. In particular, the zero lengths do not influence the aspect ratio.}
The techniques here use the so-called low diameter decomposition to define a frequency function that sparsifies the rounding/scaling DP table.
Crucial to the completeness of the sparsified DP is to run it on the input graph with a so-called sparse approximate hopset, the best known of which has complexity $O(n^{1/3})$ (see \cite{bernstein2022closing}, which includes the definition of hopsets), and for which any construction has complexity $\Omega(n^{1/4})$~\cite{BodwinH23}.
Without explaining the details, this means that pushing the techniques of \cite{ashvinkumar2025faster} to their limits would yield an $\Ot_\eps(mn^{1/4})$ time algorithm at best.
This is the \emph{hopset barrier}, which even more recent work\footnote{Personal correspondence.} approaches, giving a randomized $\Ot_\eps(m\sqrt{n}\log(nW))$ time $(1+\eps,1+\eps)$-approximate algorithm by using, surprisingly, techniques from parallel algorithms to sequentially construct a length-cost hopset.
At any rate, completely new ideas are needed to break the hopset barrier.

\subsection*{Our Results}

The main result in this paper is a \emph{deterministic} and \emph{near-linear} time $(1+\eps,1+\eps)$-approximate algorithm for $\rsp$, resolving the complexity of $\rsp$ up to logarithmic factors.
\begin{theorem}
\label{thm:main}
    There is a deterministic $(1+\eps, 1+\eps)$-approximate algorithm for single-source $\rsp$ that runs in $\Ot(m\log(nW)/\eps^9)$ time, where $W$ is the aspect ratio of lengths.
\end{theorem}

Not only are our results nearly optimal, they come with the added bonus of being very natural and without heavy machinery.
The result is entirely self-contained.
Moreover, there is only one $\log n$ factor being suppressed.
A routine modification of the ideas supplying \Cref{thm:main} yields an analogous algorithm for $\rsp$ with $k$ resources.

\begin{theorem}
\label{thm:general}
    For any fixed value of $k \ge 2$, let $\eps_1 = \ldots = \eps_k = \eps$.
    There is a deterministic $(1+\eps_1, \ldots, 1+\eps_k)$-approximate algorithm for single-source $\rsp$ with $k$ resources that runs in $\Ot(m\log(nW)/\eps^{4k+1})$ time, where $W$ is the aspect ratio of lengths.
\end{theorem}

While it remains open to give a strongly polynomial near-linear time algorithm (in fact, sub-$mn$ time), standard techniques applied to our results give strongly polynomial near-linear time algorithms for the single-source single-target variant (see \cite{ashvinkumar2025faster}).
For expositional clarity, we did not try to optimize the dependence on $1/\varepsilon$; it can certainly be improved.

\subsection*{Organization and Presentation}

The focus of this paper is primarily to convey the \emph{ideas} which supply \Cref{thm:main}.
For this reason, we omit the details of \Cref{thm:general}; they do not involve any new ideas given those of our main result.
A simplifying assumption we make in this paper is that lengths and costs are non-zero; this can be easily fixed by defining a $\mathrm{Pois}(0)$ random variable to be deterministically $0$, and doing some minor extra bookkeeping.
We also assume $\eps < 1$.
To keep the main ideas in focus, we present an algorithm which makes access to infinite precision real numbers (i.e. real RAM computation).
Truncating Poisson probabilities to $O(\log(n/\eps))$ bits of precision allows for implementation in the word RAM with a polylogarithmic running time overhead.
The paper is organized as follows:
\begin{itemize}
    \item Preliminaries are in \Cref{sec:preliminaries}.
    \item The high-concept explanation of our new ideas is provided in \Cref{sec:overview}.
    \item The formal algorithm is in \Cref{sec:algorithm}.
        This is a deterministic algorithm that outputs distances for each $t \in V$, and also outputs a data structure of size $\Ot_\eps(m \log(nW))$ from which witness paths can be retrieved \emph{randomly}.
    \item The analysis of the algorithm is in \Cref{sec:analysis}.
    \item While distances and the representation of paths are deterministically constructed, the random retrieval of paths suggests we have told a half-truth in claiming a deterministic algorithm.
        We show how to derandomize path retrieval in \Cref{sec:derandomization}.
\end{itemize}

\section{Preliminaries}
\label{sec:preliminaries}

\paragraph{Numbers and Sets. }
For a positive integer $n \in \IN$, we denote $\set{1, 2, \ldots, n}$ with $[n]$.
We use $[n]_0$ to denote $[n] \cup \set{0}$.
$e \approx 2.718$ is Euler's number (never an edge!).

\paragraph{Graphs. }
Throughout the paper, we consider directed graphs $G = (V, E)$.
When it is clear, we use $n = \card{V}$ and $m = \card{E}$.
We denote an edge directed from vertex $u$ to vertex $v$ with $uv$.
A path $P$ is a sequence of vertices $v_1, \ldots, v_h$ such that $v_i v_{i+1}$ are edges for all $i \in [h - 1]$; we call $P$ a $(v_1,v_h)$-path and we sometimes treat it as the (multi-) set of its edges.
We concatenate a vertex $v$ to the start of $P$ by $v \cdot P$.
If we write $s \in V$ or $t \in V$ and do not qualify them further, they stand for ``source'' and ``target'', respectively.

The graphs in this paper will come equipped with lengths $\ell: E \rightarrow \IR_{> 0}$ and costs
${c: E \rightarrow \IR_{> 0}}$.
The lengths and costs of a path $P$ is written with $\ell(P) = \sum_{uv \in P} \ell(uv)$ and $c(P) = \sum_{uv \in P} c(uv)$ respectively.

\paragraph{Probability. }
We denote by $X \sim \pois(\lambda)$ if $X$ is a Poisson random variable with mean $\lambda$ and $Y \sim \expod(\lambda)$ if $Y$ is an exponential random variable with rate $\lambda$.
See \cite{mitzenmacher2017probability} for more on these.

\section{Key Ideas and Motivation}
\label{sec:overview}

In this section we describe the main ideas behind our result.
For clarity of exposition, we sacrifice some mathematical precision in the ensuing discussion.

\subsection*{The Sum Trick and the Shared Sum Stalemate}

Consider the following simple near-linear time $(2+\eps,2)$-approximate algorithm for $\rsp$.
\begin{tbox}
    For $O(\log(nW)/\eps)$ guesses $L = (1+\eps)^i$:
    \begin{enumerate}
        \item $\ell_i \gets \ell/L$.
        \item $c_i \gets c/C$.
        \item Run Dijkstra's algorithm from $s$ on $G$ with edge weights $\ell_i + c_i$.
    \end{enumerate}
    The distance estimate of $t \in V$ is the smallest $L=(1+\eps)^i$ for which the Dijkstra execution found $t$ to have length-plus-cost distance $d_i(t) \le 2$.
\end{tbox}
The length-cost shared sum was shown to be a meaningful quantity in \cite{ashvinkumar2025faster}, and here we can see why by noting:
\begin{align*}
    &\textrm{(Completeness) If there is an } (s,t)\textrm{-path } P \textrm{ with } \ell(P) \le L \textrm{ and } c(P) \le C \textrm{ then } d_i(t) \le 2.
    \\
    &\textrm{(Soundness)  If } d_i(t) \le 2 \textrm{ then there is an } (s,t)\textrm{-path } P \textrm{ with } \ell(P) \le 2(1+\eps)^i \textrm{ and } c(P) \le 2C.
\end{align*}
It is also worth noting the following nice feature of this algorithm: witness paths are represented by a data structure of size $\Ot_\eps(n \log(nW))$, analogously to how shortest path algorithms produce a shortest path tree.
Indeed, if the distance estimate of $t$ is $L$, then its witness path is the $(s,t)$-path in the shortest path tree produced by the Dijkstra execution for guess $L$.
Some form of this will show up in our main result.

At any rate, a natural thing to try is to drive one of the approximation factors down by a weighted summation.
\begin{tbox}
    For $O(\log(nW)/\eps)$ guesses $L = (1+\eps)^i$:
    \begin{enumerate}
        \item $\ell_i \gets \ell/(\textcolor{red}{\alpha} L)$.
        \item $c_i \gets c/(\textcolor{red}{\beta} C)$.
        \item Run Dijkstra's algorithm from $s$ on $G$ with edge weights $\ell_i + c_i$.
    \end{enumerate}
    The distance estimate of $t \in V$ is the smallest $L=(1+\eps)^i$ for which the Dijkstra execution found $t$ to have length-plus-cost distance $d_i(t) \le \textcolor{red}{1/\alpha + 1/\beta}$.
\end{tbox}
Unfortunately, this kind of reweighting yields a $(1+\alpha/\beta+\eps,1+\beta/\alpha)$-approximate algorithm.
We can drive the length (resp. cost) approximation down arbitrarily close to $1+\eps$, but this is at the expense of the cost (resp. length) approximation being driven up!

\subsection*{Breaking the Shared Sum Stalemate with Exploding Random Clocks}

It turns out that the weighted shared sum is not the only way to drive an approximation factor down.
There is in fact a much better way, which brings \emph{both} approximation factors down.
To see this, we first give a \emph{probabilistic interpretation} to the $(2+\eps,2)$-approximate algorithm.
Assume for simplicity $L = (1+\eps)^0 = 1$ and $C = 1$, which in effect means that there is no need to normalize $\ell$ and $c$ in the $(2+\eps,2)$-approximate algorithm; we just run Dijkstra's algorithm with a cutoff distance of $2$.
Now consider the following scenario:
\begin{itemize}
    \item You hold on to two independent exponential clocks called the length-clock and cost-clock.
        The length-clock and cost-clock explode at times $\lclock$ and $\cclock$ respectively, which are independent random variables with $\lclock, \cclock \sim \expod(1)$.
    \item You can walk backwards across an edge $uv$, from $v$ to $u$.
        Doing so will advance time underlying $\lclock$ by $\ell(uv)$ and time underlying $\cclock$ by $c(uv)$.
    \item You stand initially upon $t$.
        The goal is to get back home to $s$ safely (i.e. with not a single explosion).
\end{itemize}
By taking an $(s,t)$-path $P$ (recall, backwards from $t$ to $s$), the probability of a deadly trip home is thus $1 - \expo\left(-(\ell(P) + c(P))\right)$.
This is just the length-cost shared sum weight of $P$, passed through the function $1 - \expo(-x)$.
What, then, does $d_0(t)$ mean?
This is the $\left(1 - \expo(-x)\right)$-rescaled probability of extinction by selecting an optimal route from $t$ to $s$!
Thus, under this probabilistic lens, we accept all $t \in V$ which have extinction probability $\le 1 - e^{-2}$.

A conceptual reason why this algorithm yields $(2+\eps,2)$-approximate solutions is that taking the three distinct paths\footnote{We temporarily discharge the assumption that lengths and costs are non-zero for this example.} $P_{1,1}$ and $P_{2,0}$ and $P_{0,2}$, where $\ell(P_{x,y}) = x$ and $c(P_{x,y}) = y$, taking each of these paths yields the exact same survival probability of $e^{-2}$.
If $P_{1,1}$ is indeed exactly optimal, $P_{2,0}$ and $P_{0,2}$ are $(2,2)$-approximate.
It stands to reason that increasing the survival probability of taking $P_{1,1}$ and decreasing the survival probability of taking $P_{2,0}$ and/or $P_{0,2}$ would result in a tighter approximation.
To achieve this, we concentrate the clocks further on their mean of $1$.

The time-to-explode for a clock can be modeled as a single tick of the clock.
We can increase concentration by modeling the time-to-explode with a two-tick countdown, with both the first and second ticks being independently distributed according to $\expod(2)$ where the second tick's timer starts after the first tick.
Here, the mean time for explosion (that is, the second tick) remains $1$.
Importantly, note how the survival probabilities of the three paths change.
A calculation shows the survival probability of taking $P_{1,1}$ to be $9e^{-4}$ and $P_{2,0}$ and $P_{0,2}$ to be $5e^{-4}$.
Changing the extinction probability cutoff to $\le 1 - 9e^{-4}$ (instead of the original cutoff of $\le 1 - e^{-2}$) keeps $P_{1,1}$ but rules out $P_{2,0}$ and $P_{0,2}$ as escape routes.

By further concentrating times-to-explode to within a $\Theta(\eps)$ width window around their mean of $1$ (we will need an $\Omega(\eps^{-2})$-tick countdown to achieve this), we arrive at a $(1+\eps,1+\eps)$-approximate algorithm.
Having $k$ independent clocks gets the version for $\rsp$ with $k$ resources (see \Cref{sec:k-resources}).

\section{The Algorithm}
\label{sec:algorithm}

In this section we describe the algorithm which attains \Cref{thm:main}.
Analysis of this algorithm is deferred to \Cref{sec:analysis}.
We follow the same outer framework of the $(2+\eps,2 )$-approximate algorithm described in \Cref{sec:overview}, where $L$ is guessed to be some power of $(1+\eps)$.
Henceforth, we fix $L$ in this section.
The algorithm we formalize in this section decides a set of vertices satisfying the following conditions (which will be proved in \Cref{sec:analysis}):
\begin{align*}
    &\textrm{(Completeness) If there is an } (s,t)\textrm{-path } P \textrm{ with } \ell(P) \le L \textrm{ and } c(P) \le C \textrm{ then } t \textrm{ is in the set.}
    \\
    &\textrm{(Soundness)  If } t \textrm{ is in the set then there is an } (s,t)\textrm{-path } P \textrm{ with } \ell(P) \le (1+\eps)L \textrm{ and } c(P) \le (1+\eps)C.
\end{align*}
The algorithm also outputs a table from which we can retrieve such witness paths, not unlike how a shortest path algorithm outputs a shortest path tree.
By taking the smallest value of $L$ for which $t$ is accepted, we can retrieve a $((1+\eps)^2, 1+\eps)$-approximate $(s,t)$-path from the table constructed for that $L$.
By initially rescaling $\eps$ down by a sufficiently large constant factor, this is a $(1+\eps', 1+\eps')$-approximate $(s,t)$-path.
We now describe the inner framework, which generalizes the Dijkstra execution of the $(2+\eps,2)$-approximate algorithm to multiple ticks.

\paragraph{Parameters.}
We model the length-clock and cost-clock to explode on the $q$th tick, where $q = \Theta(\eps^{-2})$.
Let $\llambda = \frac{q}{(1+\eps/2)L}$ and $\clambda = \frac{q}{(1+\eps/2)C}$.
Finally, for each edge $uv \in E$, define the independent random variables $L_{uv} \sim \pois(\llambda \ell(uv))$ and $C_{uv} \sim \pois(\clambda c(uv))$; these should not be confused with the fixed values $L$ and $C$, which are not random.

\paragraph{Inner Algorithm.}
We build a DP table with $O(nq^2)$ states of the form $(t,i,j)$ where $t \in V$ and $i,j \in [q]_0$.
The state $(t,i,j)$ is meant to store the minimum extinction probability starting at vertex $t$ when the length-clock has already ticked exactly $i$ times and the cost-clock has already ticked exactly $j$ times.
We initialize the boundary states: $\dpt(v,i,q) \gets 1$ for all $v \in V$ and all $i \in [q]_0$; $\dpt(v,q,j) \gets 1$ for all $v \in V$ and all $j \in [q]_0$; and $\dpt(s,i,j) \gets 0$ for all $i < q$ and $j < q$.
The DP table also supplies pointers from states to the vertex index (i.e. first index) of parent states.
This pointer is viewed as an instruction for which vertex to go to next -- in totality, an escape plan --, which is the one that minimizes extinction probability.
The inner algorithm is then formalized thusly.

\begin{tbox}
    \algname{$\rsp$ Inner Algorithm (Escape Plan Construction)}\\
    \textbf{Input:} An $n$-vertex $m$-edge directed graph $G=(V,E)$ with edge lengths $\ell(\cdot)$ and edge costs $c(\cdot)$, a source $s \in V$, a length threshold $L$, a cost threshold $C$.\\
    \textbf{Output:} A set of vertices, and a DP table.
    \begin{enumerate}
        \item Initialize DP table boundary values.
        \item For $i$ ranging down from $q-1$ to $0$ and $j$ ranging down from $q-1$ to $0$:
            \begin{enumerate}
                \item Initialize a min-priority queue $\cQ$.
                \item Add $s$ to $\cQ$ with key $0$.
                \item Add $v$ to $\cQ$ with key $\min(\dpt(v,i+1,j), \dpt(v,i,j+1))$ for all $v \in V \setminus \set{s}$.
                    \\$\dpt(v,i,j)$ points to $v$ with an additional label \texttt{length} (if argmin is $\dpt(v,i+1,j)$) or \texttt{cost} (otherwise).
                \item While $\cQ$ is not empty:
                    \begin{enumerate}
                        \item Let $u$ be the popped element from $\cQ$ and set $\dpt(u,i,j)$ to be its key.
                        \item For all edges $uv$ ($v$ still in $\cQ$), update the key of $v$ to
                        \begin{align*}
                            &\prob{i + L_{uv} \ge q \textrm{ or } j + C_{uv} \ge q}
                            \\+
                            &\sum_{\substack{k_\ell \in [q-i-1]_0\\k_c \in [q-j-1]_0}} \prob{L_{uv} = k_\ell \textrm{ and } C_{uv} = k_c} \dpt(u,i+k_\ell,j+k_c)
                            \tag{\ding{67}}
                        \end{align*}
                            if strictly smaller than the current key.
                            \\$\dpt(v,i,j)$ points to $u$ if key is updated.
                    \end{enumerate}
            \end{enumerate}
        \item Accept all $t$ with $\dpt(t,0,0) \le 1/4$.
    \end{enumerate}
\end{tbox}
The first term in (\ding{67}), the key update expression, is the probability of explosion while walking (backwards) on $uv$.
The second term in (\ding{67}) calculates the extinction probabilities of disjoint events where $k_\ell$ and $k_c$ ticks elapse while crossing $uv$.

\section{Analysis}
\label{sec:analysis}

In this section we prove the correctness and time complexity of the $\rsp$ inner algorithm and thence establish \Cref{thm:main}.
First, note that
\begin{align*}
    \dpt(t,i,j) = &
    \ \prob{\begin{array}{l}\textrm{Length-clock or cost-clock explodes}\\\textrm{when following an optimal adaptive escape plan from } t\\\textrm{with respective countdowns } q-i \textrm{ and } q-j \end{array}}
    \\
    \le
    &
    \min_{(s,t)\textrm{-path } P} \prob{\begin{array}{l}\textrm{Length-clock or cost-clock explodes}\\\textrm{with respective countdowns } q-i \textrm{ and } q-j \\\textrm{when taking }P\end{array}},
\end{align*}
where the inequality holds since a fixed $P$ is a non-adaptive escape plan.
One can see the faithfulness by using a simple argument on the execution of the $\rsp$ inner algorithm in the same way one traces the execution of Dijkstra to prove its correctness (see \Cref{sec:dijkstra}).

\subsection*{Completeness}

We now show that there are no false-negative targets which can be reached by $s$.
\begin{lemma}
\label{lem:complete}
    Suppose there is an $(s,t)$-path $P$ with $\ell(P) \le L$ and $c(P) \le C$.
    Then $t$ is accepted by the $\rsp$ inner algorithm.
\end{lemma}
\begin{proof}
    Consider walking backwards along such a path $P$, exposing the randomness of $L_{uv}$ and $C_{uv}$ as we step across $uv \in P$.
    We have
    \begin{align*}
        &\sum_{uv \in P} L_{uv} \sim \pois(\llambda \ell(P)),
        \\
        &\sum_{uv \in P} C_{uv} \sim \pois(\clambda c(P)),
    \end{align*}
    and so, by standard bounds (with $q = \Theta(\eps^{-2})$) and defining $\lclock \sim \pois(\llambda L)$ and $\cclock \sim \pois(\clambda C)$,
    \begin{align*}
        &\prob{\sum_{uv \in P} L_{uv} \ge q} \le \prob{\lclock \ge q} \le 1/8
        \\
        &\prob{\sum_{uv \in P} C_{uv} \ge q} \le \prob{\cclock \ge q} \le 1/8.
    \end{align*}
    The probability of extinction walking along $P$ is therefore at most $1/4$ by a union bound, and $\dpt(t,0,0)$ is no larger.
\end{proof}

\subsection*{Soundness}

Before proving soundness, consider the following random process to produce an $(s,t)$-path.
\begin{tbox}
    \algname{Random Path Retrieval (Following the Escape Plan)}\\
    \textbf{Input:} Escape plan DP table, $t$.\\
    \textbf{Output:} An $(s,t)$-path $P$ (or explosion \explode).
    \begin{enumerate}
        \item $P \gets t$.
        \item $\lclockt \gets 0$ and $\cclockt \gets 0$.
            These track ticks so far.
        \item While $v \neq s$ where $v$ is the first vertex of $P$:
            \begin{enumerate}
                \item $u \gets $ parent of $\dpt(v,\lclockt,\cclockt)$.
                \item If $u = v$:
                    \begin{enumerate}
                        \item If $\dpt(v,\lclockt,\cclockt)$'s label is \texttt{length}: $\lclockt \gets \lclockt + 1$.
                        \item If $\dpt(v,\lclockt,\cclockt)$'s label is \texttt{cost}: $\cclockt \gets \cclockt + 1$.
                    \end{enumerate}
                \item Else:
                    \begin{enumerate}
                        \item $P \gets u \cdot P$.
                        \item $\lclockt \gets \lclockt + L_{uv}$ using fresh randomness.
                        \item $\cclockt \gets \cclockt + C_{uv}$ using fresh randomness.
                    \end{enumerate}
                \item If $\lclockt \ge q$ or $\cclockt \ge q$: Return \explode.
            \end{enumerate}
        \item Return $P$.
    \end{enumerate}    
\end{tbox}

We now show there are no false-positive targets which cannot be reached by $s$.
\begin{lemma}
\label{lem:sound}
    Suppose $\dpt(t,0,0) \le 1/4$.
    Then there is an $(s,t)$-path $P$ with $\ell(P) \le (1+\eps)L$ and $c(P) \le (1+\eps)C$.
\end{lemma}
\begin{proof}
    Consider the random path retrieval process from $t$.
    We show that with positive probability, an $(s,t)$-path $P$ with $\ell(P) \le (1+\eps)L$ and $c(P) \le (1+\eps)C$ is returned.
    Define $\lclock \sim \pois(\llambda (1+\eps)L)$ and $\cclock \sim \pois(\clambda (1+\eps)C)$.
    The probability that this does not happen is
    \begin{align*}
        \prob{\begin{array}{l}\textrm{\explode returned}\\\textrm{or }P\textrm{ returned with }\ell(P)>(1+\eps)L \\\textrm{or }P\textrm{ returned with }c(P)>(1+\eps)C\end{array}}
        &\le \begin{array}{l}\phantom{+}\prob{\textrm{\explode returned}} \\+\prob{P\textrm{ returned with }\ell(P)>(1+\eps)L} \\+\prob{P\textrm{ returned with }c(P)>(1+\eps)C}\end{array}
        \tag{Union bound}
        \\
        &\le \begin{array}{l}\phantom{+}1/4 \\+\prob{P\textrm{ returned with }\ell(P)>(1+\eps)L} \\+\prob{P\textrm{ returned with }c(P)>(1+\eps)C}\end{array}
        \tag{$\dpt(t,0,0) \le 1/4$}
        \\
        &\le 1/4 + \prob{\lclock < q} + \prob{\cclock < q}
        \tag{Stochastic dominance}
        \\
        & < 1/2,
    \end{align*}
    where the last inequality follows from standard bounds on the Poisson distribution (with $q = \Theta(\eps^{-2})$), and the stochastic dominance claim holds by coupling the $(1+\eps)L$ length (resp. $(1+\eps)C$ cost) suffix of the random output path $P$ with a Poisson point process with rate $\llambda$ (resp. $\clambda$) over an interval of size $(1+\eps)L$ (resp. $(1+\eps)C$).
    Consequently, a $(1+\eps,1+\eps)$-approximate path $P$ must exist.
\end{proof}
Worth noting in the above proof is that one can deduce that we need only try the random path retrieval process $\Theta(\log n)$ times to successfully extract a $(1+\eps,1+\eps)$-approximate path with high probability.

\subsection*{Putting Everything Together}
Finally, we can give the proof of our main result.
\begin{proof}[Proof of \Cref{thm:main}]\ \\
    \underline{Correctness}:
    By scaling $\eps \gets \eps/4$, we have $(1+\eps_{\textrm{new}})^2 \le 1+\eps_{\textrm{old}}$ so it suffices to show a $((1+\eps)^2,1+\eps)$-approximate algorithm.
    \Cref{lem:complete} and \Cref{lem:sound} together, for an arbitrarily precise value of $L$, yield a $(1+\eps,1+\eps)$-approximate algorithm.
    Since $L$ has a resolution of $(1+\eps)$, the combined framework of guessing and running the $\rsp$ inner algorithm gives a $((1+\eps)^2,1+\eps)$-approximate algorithm.
    \\
    \underline{Time Complexity}:
    The outer loop of the $\rsp$ inner algorithm iterates $q^2 = O(\eps^{-4})$ times.
    Each iteration runs Dijkstra's algorithm, but the candidate keys take $q^2$ time to compute since that is the number of summands.
    The inner algorithm thus takes $\Ot(m/\eps^{8})$ time.
    This is executed $O(\log(nW)/\eps)$ times by the outer framework, completing the proof.
\end{proof}

\section{Derandomization of Path Retrieval}
\label{sec:derandomization}

In this section we show how witness paths can be deterministically retrieved from the final $O_\eps(\log nW)$ DP tables, each of size $\Ot_\eps(m)$.
Our approach is to use the method of conditional expectations on the random path retrieval procedure, but caution must be taken.
The randomness is clearly exposed by the $L_{uv}$'s and $C_{uv}$'s, but what score are we taking the expectation of?
Naively, we could try to use $\dpt(u,i,j)$ -- the extinction probability of the optimal escape plan -- as a score which we try to minimize.
But this incentivizes the derandomization to always choose zero ticks for the length-clock and cost-clock!
We need to take the history of the walk so far into account in order to correct the unlikely situation where lots of length or cost has been spent with a disproportionately low number of clock ticks.
In view of this, we define the potential function $\Phi$.

\paragraph{Potential. }
Let $z = \frac{1+\eps/2}{1+\eps}$.
\begin{align*}
    \Phi(u,i,j,P) = \dpt(u,i,j) + \frac{z^i}{8\prod\limits_{xy \in P}\expect{z^{L_{xy}}}} + \frac{z^j}{8\prod\limits_{xy \in P}\expect{z^{C_{xy}}}}.
\end{align*}
We will use the method of conditional expectations to minimize $\Phi$.

\paragraph{Path Retrieval.}
We are now ready to state the procedure for deterministically retrieving witness paths from the DP table.

\begin{tbox}
    \algname{Deterministic Path Retrieval  (Following the Escape Plan)}\\
    \textbf{Input:} Escape plan DP table, $t$.\\
    \textbf{Output:} An $(s,t)$-path $P$.
    \begin{enumerate}
        \item $P \gets t$.
        \item $\lclockt \gets 0$ and $\cclockt \gets 0$.
            These track ticks so far.
        \item $\lH \gets 1/8$ and $\cH \gets 1/8$.
            \\These track the second and third summands of $\Phi(v,\lclockt,\cclockt,P)$ where $v$ is the first vertex of $P$.
        \item While $v \neq s$ where $v$ is the first vertex of $P$:
            \begin{enumerate}
                \item $u \gets $ parent of $\dpt(v,\lclockt,\cclockt)$.
                \item If $u = v$:
                    \begin{enumerate}
                        \item If $\dpt(v,\lclockt,\cclockt)$'s label is \texttt{length}: $\lclockt \gets \lclockt + 1$ and $\lH \gets z\lH$.
                        \item If $\dpt(v,\lclockt,\cclockt)$'s label is \texttt{cost}: $\cclockt \gets \cclockt + 1$ and $\cH \gets z\cH$.
                    \end{enumerate}
                \item Else:
                    \begin{enumerate}
                        \item $P \gets u \cdot P$.
                        \item $I \gets [q-\lclockt-1]_0\times[q-\cclockt-1]_0$.
                        \item \begin{align*}
                                (i_\ell, i_c) \gets \operatorname*{arg\,min}_{(i_\ell, i_c) \in I}\left(\begin{array}{l}\phantom{+}\dpt(u,\lclockt + i_\ell,\cclockt + i_c) \\+ z^{i_\ell}\lH / \expect{z^{L_{uv}}} \\+ z^{i_c}\cH / \expect{z^{C_{uv}}} \end{array}\right).
                            \end{align*}
                        \item $\lclockt \gets \lclockt + i_\ell$.
                        \item $\cclockt \gets \cclockt + i_c$.
                        \item $\lH \gets z^{i_\ell}\lH / \expect{z^{L_{uv}}}$.
                        \item $\cH \gets z^{i_c}\cH / \expect{z^{C_{uv}}}$.
                    \end{enumerate}
            \end{enumerate}
        \item Return $P$.
    \end{enumerate}    
\end{tbox}

Before we bound the length and cost of the output path $P$, we first show that a violation gives a signal in $\Phi$.

\begin{observation}
\label[observation]{obs:violation}
    Let $P$ be a path output by the deterministic path retrieval procedure.
    \begin{itemize}
        \item $\ell(P) > (1+\eps)L \textrm{ and } \lclockt < q \implies \Phi \ge 1$,
        \item $c(P) > (1+\eps)C \textrm{ and } \cclockt < q \implies \Phi \ge 1$.
    \end{itemize}
\end{observation}
\begin{proof}
    We prove the first item (the second is analogous).
    It suffices to show $\frac{z^{\lclockt}}{\prod\limits_{uv \in P}\expect{z^{L_{uv}}}} \ge 8$, where the left hand side comes from the second summand of $\Phi$.
    \begin{align*}
        \frac{z^{\lclockt}}{\prod\limits_{uv \in P}\expect{z^{L_{uv}}}} 
        \ge\ &
        \frac{z^{q-1}}{\prod\limits_{uv \in P}\expect{z^{L_{uv}}}}
        \tag{$z < 1$}
        \\
        =\ & 
        z^{q-1}e^{(1-z)\llambda \ell(P)}
        \tag{PGF of Poisson distribution}
        \\
        >\ &
        z^{q-1}e^{(1-z)\llambda (1+\eps)L}
        \tag{$\ell(P) > (1+\eps)L$}
        \\
        =\ & 
        z^{q-1}e^{(1-z)q/z} \ge 8.
        \tag{$\llambda = \frac{q}{(1+\eps/2)L}$}
    \end{align*}
    The last inequality follows from standard bounds (with $q = \Theta(\eps^{-2})$ and $z = (1+\eps/2)/(1+\eps)$).
\end{proof}

We are finally ready to show length and cost bounds on the witness path extracted from the deterministic path retrieval procedure.

\begin{lemma}
    Suppose $\dpt(t,0,0) \le 1/4$.
    Then the deterministic path retrieval procedure produces an $(s,t)$-path $P$ with $\ell(P) \le (1+\eps)L$ and $c(P) \le (1+\eps)C$.
\end{lemma}
\begin{proof}\ \\
    \underline{Bounding the potential}:
    The initial potential at the start of the retrieval procedure is $\Phi(t,0,0,t) = \dpt(t,0,0) + \lH + \cH \le 1/4 + 1/8 + 1/8 < 1$.
    
    We will show that the potential stays below $1$.
    For this purpose, extend the DP table infinitely beyond its boundary, with values $1$ at these new cells.
    Assume
    \begin{align*}
        \Phi(v,\lclockt,\cclockt,P) = \dpt(v,\lclockt,\cclockt) + \frac{z^{\lclockt}}{8\prod\limits_{xy \in P}\expect{z^{L_{xy}}}} + \frac{z^{\cclockt}}{8\prod\limits_{xy \in P}\expect{z^{C_{xy}}}} < 1,
    \end{align*}
    and let $u$ be the parent of $v$.
    We do a casework on $\set{v = u}, \set{v \neq u}$ (the second case is more interesting).
    \begin{itemize}
        \item If $v = u$, then suppose without loss of generality that the label is \texttt{length}.
            Then,
            \begin{align*}
                &\Phi(u, \lclockt + 1, \cclockt, P)
                \\
                =\ &
                \dpt(u,\lclockt + 1,\cclockt) + z \frac{z^{\lclockt}}{8\prod\limits_{xy \in P}\expect{z^{L_{xy}}}} + \frac{z^{\cclockt}}{8\prod\limits_{xy \in P}\expect{z^{C_{xy}}}}
                \\
                <\ &
                \dpt(u,\lclockt + 1,\cclockt) + \frac{z^{\lclockt}}{8\prod\limits_{xy \in P}\expect{z^{L_{xy}}}} + \frac{z^{\cclockt}}{8\prod\limits_{xy \in P}\expect{z^{C_{xy}}}}
                \tag{$z < 1$}
                \\
                =\ &
                \dpt(u,\lclockt,\cclockt) + \frac{z^{\lclockt}}{8\prod\limits_{xy \in P}\expect{z^{L_{xy}}}} + \frac{z^{\cclockt}}{8\prod\limits_{xy \in P}\expect{z^{C_{xy}}}}
                \tag{Definition of DP update rule}
                \\
                =\ & \Phi(u,\lclockt,\cclockt,P) < 1,
                \tag{By assumption and $u = v$}
            \end{align*}
            as claimed.
        \item If, on the other hand, $v \neq u$, then
            \begin{align*}
                &\expect{\Phi(u,\lclockt + L_{uv},\cclockt + C_{uv},u \cdot P)}
                \\
                =\ &
                \expect{\dpt(u,\lclockt + L_{uv},\cclockt+ C_{uv})}
                + \expect{\frac{z^{\lclockt + L_{uv}}}{8\prod\limits_{xy \in u\cdot P}\expect{z^{L_{xy}}}}}
                + \expect{\frac{z^{\cclockt + C_{uv}}}{8\prod\limits_{xy \in u\cdot P}\expect{z^{C_{xy}}}}}
                \\
                =\ &
                \expect{\dpt(u,\lclockt + L_{uv},\cclockt+ C_{uv})}
                + \frac{z^{\lclockt}}{8\prod\limits_{xy \in P}\expect{z^{L_{xy}}}} 
                + \frac{z^{\cclockt}}{8\prod\limits_{xy \in P}\expect{z^{C_{xy}}}}
                \\
                =\ &
                \dpt(v,\lclockt,\cclockt)
                + \frac{z^{\lclockt}}{8\prod\limits_{xy \in P}\expect{z^{L_{xy}}}} 
                + \frac{z^{\cclockt}}{8\prod\limits_{xy \in P}\expect{z^{C_{xy}}}}
                \tag{Definition of DP update rule}
                \\
                =\ &
                \Phi(v,\lclockt,\cclockt,P) < 1.
                \tag{By assumption}
            \end{align*}
            There must therefore exist values $L_{uv} = i_\ell$ and $C_{uv} = i_c$ for which $\Phi(u,\lclockt + i_\ell,\cclockt + i_c ,u \cdot P) < 1$.
            The value of $(i_\ell, i_c)$ must necessarily be in $[q-\lclockt-1]_0\times[q-\cclockt-1]_0$ or else the first summand of the potential is $1$ (recall that we extended the DP table infinitely).
    \end{itemize}

    \underline{Establishing termination (some path $P$ is output)}:
    Finally, note that the procedure terminates and outputs a path $P$ since each step either simulates a clock tick (of which there can be no more than $2q$ total increments) or updates the front of the path to a vertex which was extracted earlier from the min-priority queue by the Dijkstra execution on the $\rsp$ inner algorithm.

    \underline{Putting things together}:
    Hence, by using \Cref{obs:violation} with $\Phi < 1$ to bound $\ell(P)$ and $c(P)$, the claim follows.
\end{proof}

\ 
\\\\\textbf{Acknowledgments. }
I thank Aaron Bernstein and Adam Karczmarz for the many illuminating discussions we had while working on \cite{ashvinkumar2025faster}.

\bibliographystyle{alpha}
\bibliography{references}

@inproceedings{ashvinkumar2025faster,
  author       = {Ashvinkumar, Vikrant and Bernstein, Aaron and Karczmarz, Adam},
  title        = {Faster approximation algorithms for restricted shortest paths in directed graphs},
  booktitle    = {Proceedings of the 2025 Annual {ACM-SIAM} Symposium on Discrete Algorithms ({SODA})},
  pages        = {5263--5277},
  publisher    = {SIAM},
  year         = {2025},
  doi          = {10.1137/1.9781611978322.180},
  url          = {https://arxiv.org/abs/2410.17179},
  note         = {\href{https://arxiv.org/abs/2410.17179}{\mbox{arXiv:2410.17179}}}
}

@inproceedings{Bernstein12,
  author       = {Bernstein, Aaron},
  title        = {Near linear time {$(1+\epsilon)$}-approximation for restricted shortest paths in undirected graphs},
  booktitle    = {Proceedings of the Twenty-Third Annual {ACM-SIAM} Symposium on Discrete Algorithms ({SODA})},
  pages        = {189--201},
  publisher    = {SIAM},
  year         = {2012},
  doi          = {10.1137/1.9781611973099.17},
  url          = {https://doi.org/10.1137/1.9781611973099.17}
}

@inproceedings{BodwinH23,
  author       = {Bodwin, Greg and Hoppenworth, Gary},
  title        = {Folklore sampling is optimal for exact hopsets: Confirming the {$\sqrt{n}$} barrier},
  booktitle    = {2023 {IEEE} 64th Annual Symposium on Foundations of Computer Science ({FOCS})},
  pages        = {701--720},
  publisher    = {IEEE},
  year         = {2023},
  doi          = {10.1109/FOCS57990.2023.00046},
  url          = {https://arxiv.org/abs/2304.02193},
  note         = {\href{https://arxiv.org/abs/2304.02193}{\mbox{arXiv:2304.02193}}}
}

@inproceedings{BernsteinNW22,
  author       = {Bernstein, Aaron and Nanongkai, Danupon and Wulff{-}Nilsen, Christian},
  title        = {Negative-weight single-source shortest paths in near-linear time},
  booktitle    = {2022 {IEEE} 63rd Annual Symposium on Foundations of Computer Science ({FOCS})},
  pages        = {600--611},
  publisher    = {IEEE},
  year         = {2022},
  doi          = {10.1109/FOCS54457.2022.00063},
  url          = {https://arxiv.org/abs/2203.03456},
  note         = {\href{https://arxiv.org/abs/2203.03456}{\mbox{arXiv:2203.03456}}}
}

@inproceedings{goel2001efficient,
  author       = {Goel, Ashish and Ramakrishnan, Kajamalai G. and Kataria, Deepak and Logothetis, Dimitris},
  title        = {Efficient computation of delay-sensitive routes from one source to all destinations},
  booktitle    = {Proceedings {IEEE} {INFOCOM} 2001: Twentieth Annual Joint Conference of the {IEEE} Computer and Communications Societies},
  volume       = {2},
  pages        = {854--858},
  publisher    = {IEEE},
  year         = {2001}
}

@article{lorenz2001simple,
  author       = {Lorenz, Dean H. and Raz, Danny},
  title        = {A simple efficient approximation scheme for the restricted shortest path problem},
  journal      = {Operations Research Letters},
  volume       = {28},
  number       = {5},
  pages        = {213--219},
  year         = {2001},
  doi          = {10.1016/S0167-6377(01)00069-4},
  url          = {https://doi.org/10.1016/S0167-6377(01)00069-4}
}

@article{lorenz1998qos,
  author       = {Lorenz, Dean H. and Orda, Ariel},
  title        = {{QoS} routing in networks with uncertain parameters},
  journal      = {{IEEE/ACM} Transactions on Networking},
  volume       = {6},
  number       = {6},
  pages        = {768--778},
  year         = {1998},
  doi          = {10.1109/90.748088},
  url          = {https://doi.org/10.1109/90.748088}
}

@article{younis2003constraint,
  author       = {Younis, Ossama and Fahmy, Sonia},
  title        = {Constraint-based routing in the {Internet}: Basic principles and recent research},
  journal      = {{IEEE} Communications Surveys \& Tutorials},
  volume       = {5},
  number       = {1},
  pages        = {2--13},
  year         = {2003}
}

@article{NaorST07,
  author       = {Naor, Joseph and Shachnai, Hadas and Tamir, Tami},
  title        = {Real-time scheduling with a budget},
  journal      = {Algorithmica},
  volume       = {47},
  number       = {3},
  pages        = {343--364},
  year         = {2007},
  doi          = {10.1007/s00453-006-0191-8},
  url          = {https://doi.org/10.1007/s00453-006-0191-8}
}

@article{BriskornCLLP10,
  author       = {Briskorn, Dirk and Choi, Byung-Cheon and Lee, Kangbok and Leung, Joseph Y.-T. and Pinedo, Michael L.},
  title        = {Complexity of single machine scheduling subject to nonnegative inventory constraints},
  journal      = {European Journal of Operational Research},
  volume       = {207},
  number       = {2},
  pages        = {605--619},
  year         = {2010},
  doi          = {10.1016/j.ejor.2010.05.036},
  url          = {https://doi.org/10.1016/j.ejor.2010.05.036}
}

@article{GarroppoGT10,
  author       = {Garroppo, Rosario Giuseppe and Giordano, Stefano and Tavanti, Luca},
  title        = {A survey on multi-constrained optimal path computation: Exact and approximate algorithms},
  journal      = {Computer Networks},
  volume       = {54},
  number       = {17},
  pages        = {3081--3107},
  year         = {2010},
  doi          = {10.1016/j.comnet.2010.05.017},
  url          = {https://doi.org/10.1016/j.comnet.2010.05.017}
}

@book{GareyJ79,
  author       = {Garey, Michael R. and Johnson, David S.},
  title        = {Computers and Intractability: A Guide to the Theory of {NP}-Completeness},
  publisher    = {W. H. Freeman},
  year         = {1979},
  isbn         = {0-7167-1044-7}
}

@article{joksch1966shortest,
  author       = {Joksch, Hans C.},
  title        = {The shortest route problem with constraints},
  journal      = {Journal of Mathematical Analysis and Applications},
  volume       = {14},
  number       = {2},
  pages        = {191--197},
  year         = {1966},
  doi          = {10.1016/0022-247X(66)90020-5},
  url          = {https://doi.org/10.1016/0022-247X(66)90020-5}
}

@article{Hassin92,
  author       = {Hassin, Refael},
  title        = {Approximation schemes for the restricted shortest path problem},
  journal      = {Mathematics of Operations Research},
  volume       = {17},
  number       = {1},
  pages        = {36--42},
  year         = {1992},
  doi          = {10.1287/moor.17.1.36},
  url          = {https://doi.org/10.1287/moor.17.1.36}
}

@article{Warburton87,
  author       = {Warburton, Arthur},
  title        = {Approximation of {Pareto} optima in multiple-objective, shortest-path problems},
  journal      = {Operations Research},
  volume       = {35},
  number       = {1},
  pages        = {70--79},
  year         = {1987},
  doi          = {10.1287/opre.35.1.70},
  url          = {https://doi.org/10.1287/opre.35.1.70}
}

@article{XueZTT08,
  author       = {Xue, Guoliang and Zhang, Weiyi and Tang, Jian and Thulasiraman, Krishnaiyan},
  title        = {Polynomial time approximation algorithms for multi-constrained {QoS} routing},
  journal      = {{IEEE/ACM} Transactions on Networking},
  volume       = {16},
  number       = {3},
  pages        = {656--669},
  year         = {2008},
  doi          = {10.1109/TNET.2007.900712},
  url          = {https://doi.org/10.1109/TNET.2007.900712}
}

@inproceedings{fineman2024single,
  author       = {Fineman, Jeremy T.},
  title        = {Single-source shortest paths with negative real weights in {$\widetilde{O}(mn^{8/9})$} time},
  booktitle    = {Proceedings of the 56th Annual {ACM} Symposium on Theory of Computing ({STOC})},
  pages        = {3--14},
  publisher    = {ACM},
  year         = {2024},
  url          = {https://arxiv.org/abs/2311.02520},
  note         = {\href{https://arxiv.org/abs/2311.02520}{\mbox{arXiv:2311.02520}}}
}

@inproceedings{duan2025breaking,
  author       = {Duan, Ran and Mao, Jiayi and Mao, Xiao and Shu, Xinkai and Yin, Longhui},
  title        = {Breaking the sorting barrier for directed single-source shortest paths},
  booktitle    = {Proceedings of the 57th Annual {ACM} Symposium on Theory of Computing ({STOC})},
  pages        = {36--44},
  publisher    = {ACM},
  year         = {2025},
  doi          = {10.1145/3717823.3718179},
  url          = {https://arxiv.org/abs/2504.17033},
  note         = {\href{https://arxiv.org/abs/2504.17033}{\mbox{arXiv:2504.17033}}}
}

@inproceedings{haeupler2024universal,
  author       = {Haeupler, Bernhard and Hlad{\'{\i}}k, Richard and Rozho{\v{n}}, V{\'a}clav and Tarjan, Robert E. and T{\v{e}}tek, Jakub},
  title        = {Universal optimality of {Dijkstra} via beyond-worst-case heaps},
  booktitle    = {2024 {IEEE} 65th Annual Symposium on Foundations of Computer Science ({FOCS})},
  pages        = {2099--2130},
  publisher    = {IEEE},
  year         = {2024},
  doi          = {10.1109/FOCS61266.2024.00125},
  url          = {https://arxiv.org/abs/2311.11793},
  note         = {\href{https://arxiv.org/abs/2311.11793}{\mbox{arXiv:2311.11793}}}
}

@inproceedings{bernstein2022closing,
  author       = {Bernstein, Aaron and Wein, Nicole},
  title        = {Closing the gap between directed hopsets and shortcut sets},
  booktitle    = {Proceedings of the 2023 Annual {ACM-SIAM} Symposium on Discrete Algorithms ({SODA})},
  pages        = {163--182},
  publisher    = {SIAM},
  year         = {2023},
  doi          = {10.1137/1.9781611977554.ch7},
  url          = {https://arxiv.org/abs/2207.04507},
  note         = {\href{https://arxiv.org/abs/2207.04507}{\mbox{arXiv:2207.04507}}}
}

@book{mitzenmacher2017probability,
  author       = {Mitzenmacher, Michael and Upfal, Eli},
  title        = {Probability and Computing: Randomization and Probabilistic Techniques in Algorithms and Data Analysis},
  edition      = {Second},
  publisher    = {Cambridge University Press},
  year         = {2017},
  doi          = {10.1017/9781316651124},
  isbn         = {978-1-107-15488-9}
}

\appendix
\section{Proof Sketch of DP Faithfulness}
\label{sec:dijkstra}

\begin{observation}
\label{obs:dijkstra}
    $\dpt(v,i,j) \le \expect{\dpt(u, i + L_{uv}, j+ C_{uv})}$ for all $uv \in E$, with equality when $u$ is pointed to by $\dpt(v, i, j)$.
\end{observation}
\begin{proof}[Proof sketch]
    Observe from the way $\dpt(u, x, y)$ is initialized, that $\dpt(u,x,y) \le \dpt(u,x+x',y+y')$ for all $x,y,x',y' \ge 0$ and all $u \in V$.
    This implies that the keys from $\cQ$ are popped in a non-decreasing order, since key updates are of the form $\dpt(v, x, y) \gets \sum_{x',y'} P_{uv,x',y'} \dpt(u, x+x', y+y') + P \ge \dpt(u,x,y)$ where $u$ has been popped off $\cQ$, the $P$ terms are probabilities corresponding to disjoint events summing to $1$, and $1 \ge \dpt(\cdot,\cdot,\cdot)$.
    An induction on the DP states completes the proof.
\end{proof}

\section{Sketch for Extension to $k$ Resources}
\label{sec:k-resources}
To attain \Cref{thm:general}, we let $q = A_k \eps^{-2}$ for some sufficiently large constant $A_k$.
We modify the algorithm to use $k$ clocks.
Consequently, the outer loop iterates $O(q^k)$ times, and each key update takes $O(q^k)$ time, whence the running time follows.
Completeness and soundness will now use a union bound over $\le k + 1$ summands, whose magnitudes are controlled by $A_k$.

\end{document}